\documentclass[conference]{IEEEtran}

\newif\Ifoverviewfigure

\usepackage{cite}
\usepackage{amsmath,amssymb,amsfonts}
\usepackage{graphicx}
\usepackage{textcomp}
\usepackage{xcolor}
\usepackage{amsmath}
\usepackage{amssymb}
\usepackage{amsthm}
\usepackage{mathtools}
\usepackage{graphicx}
\usepackage{listings}
\usepackage{mathabx} 
\usepackage{algorithm} 
\usepackage[noend]{algpseudocode}
\usepackage{framed}
\usepackage{booktabs} 
\usepackage{multirow} 
\usepackage{array} 
\usepackage{xfrac} 
\usepackage{wrapfig}
\usepackage{amssymb}
\usepackage{xspace}
\usepackage{url}
\usepackage{subcaption}
\usepackage{tikz}
\usepackage{arydshln}
\usetikzlibrary{trees}

\theoremstyle{plain}
\newtheorem{theorem}{Theorem}

\newtheorem{definition}{Definition}

\def\BibTeX{{\rm B\kern-.05em{\sc i\kern-.025em b}\kern-.08em
    T\kern-.1667em\lower.7ex\hbox{E}\kern-.125emX}}

\algnewcommand{\LeftComment}[1]{\Statex \(\triangleright\) #1}
\newcommand{\pparagraph}[1]{\medskip\noindent\textbf{#1}\hspace{0.5em}\xspace}
\newcommand{\eps}{\varepsilon}

\renewcommand{\epsilon}{\eps}

\algnewcommand{\IfThenElse}[3]{% \IfThenElse{<if>}{<then>}{<else>}
  \State \algorithmicif\ #1\ \algorithmicthen\ #2\ \algorithmicelse\ #3}
\algnewcommand{\IfThen}[2]{% \IfThen{<if>}{<then>}
  \State \algorithmicif\ #1\ \algorithmicthen\ #2}

\newcommand\copyrighttext{%
  \footnotesize \textcopyright 2024 IEEE. Personal use of this material is permitted.  Permission from IEEE must be obtained for all other uses, in any current or future media, including reprinting/republishing this material for advertising or promotional purposes, creating new collective works, for resale or redistribution to servers or lists, or reuse of any copyrighted component of this work in other works. DOI: \url{https://doi.org/10.1109/ICPM63005.2024.10680684}}
\newcommand\copyrightnotice{%
\begin{tikzpicture}[remember picture,overlay]
\node[anchor=south,yshift=10pt] at (current page.south) {\fbox{\parbox{\dimexpr\textwidth-\fboxsep-\fboxrule\relax}{\copyrighttext}}};
\end{tikzpicture}%
}

\begin{document}

\title{
    Differentially Private Inductive Miner
}

\author{
    \IEEEauthorblockN{Max Schulze\IEEEauthorrefmark{1}}
    \IEEEauthorblockA{\textit{Institute for IT Security} \\
    \textit{Universität zu Lübeck}\\
    max.schulze@student.uni-luebeck.de
    }
    
    \and
    \IEEEauthorblockN{Yorck Zisgen\IEEEauthorrefmark{1}}
    \IEEEauthorblockA{\textit{Chair of Business Informatics and Process Analytics} \\
    \textit{University of Bayreuth}\\
    yorck.zisgen@uni-bayreuth.de
    }
    
    \and
    \IEEEauthorblockN{Moritz Kirschte}
    \IEEEauthorblockA{\textit{Institute for IT Security} \\
    \textit{Universität zu Lübeck}\\
    m.kirschte@uni-luebeck.de
    }
    
    \and
    \IEEEauthorblockN{Esfandiar Mohammadi}
    \IEEEauthorblockA{\textit{Institute for IT Security} \\
    \textit{Universität zu Lübeck}\\
    esfandiar.mohammadi@uni-luebeck.de
    }
    
    \and
    \IEEEauthorblockN{Agnes Koschmider}
    \IEEEauthorblockA{\textit{Chair of Business Informatics and Process Analytics} \\
    \textit{University of Bayreuth, Fraunhofer FIT}\\
    agnes.koschmider@uni-bayreuth.de
    }
}

\maketitle
\copyrightnotice

\begingroup\renewcommand\thefootnote{\IEEEauthorrefmark{1}}
    \footnotetext{These authors contributed equally to this work.}
\endgroup

\begin{abstract}
    Protecting personal data about individuals, such as event traces in process mining, is an inherently difficult task since an event trace leaks information about the path in a process model that an individual has triggered. Yet, prior anonymization methods of event traces like k-anonymity or event log sanitization struggled to protect against such leakage, in particular against adversaries with sufficient background knowledge. In this work, we provide a method that tackles the challenge of summarizing sensitive event traces by learning the underlying process tree in a privacy-preserving manner. We prove via the so-called Differential Privacy (DP) property that from the resulting summaries no useful inference can be drawn about any personal data in an event trace. On the technical side, we introduce a differentially private approximation (DPIM) of the Inductive Miner. Experimentally, we compare our DPIM with the Inductive Miner on 14 real-world event traces by evaluating well-known metrics: fitness, precision, simplicity, and generalization. The experiments show that our DPIM not only protects personal data but also generates faithful process trees that exhibit little utility loss above the Inductive Miner.
\end{abstract}


\section{Introduction}
    \label{Sec:Introduction}
    Privacy risks in process mining on potentially sensitive event logs impede the valuable extraction of insights from real-world event logs, as extracted process trees can provide significant transparency into business processes. From a privacy protection perspective, however, event logs are particularly challenging: in the worst case every trace is fully associated with only one individual; hence, the footprint of an individual on an event log is very high. A study on re-identification risks in event logs showed that there are significant privacy leakages in the vast majority of the event logs used widely by the process mining community~\cite{DBLP:journals/tmis/ElkoumyFSKMVRW22,ReIDEvent}.
    
    The literature contains proposals for protecting privacy for event logs or in process mining, respectively \cite{elkoumy_differentially_2022, Fahrenkrog_Semantics-Aware,PP-ProcessM, DBLP:conf/bpm/Fahrenkrog-Petersen20,PrivPre-DFG}. Yet, four of these prior approaches \cite{elkoumy_differentially_2022,Fahrenkrog_Semantics-Aware,DBLP:conf/bpm/Fahrenkrog-Petersen20,PrivPre-DFG} do not provide strong provable privacy guarantees for a process mining algorithm against attackers with strong background knowledge as summarized in Sec.~\ref{Sec:Related_Work}. Advancements in privacy-preserving computations have demonstrated that techniques, such as k-anonymity or event log sanitization, falter when an adversary possesses sufficient background knowledge~\cite{NaSh_08:netflix}. Specifically, these methods struggle to provide substantial guarantees against future adversaries. A state-of-the-art privacy notion considering a strong attacker is differential privacy (DP), which requires that the impact of single traces on the final process model is limited; in particular, DP guarantees imply limited impact of outliers and, as a result, significantly mitigate re-identification risks. While there is work on querying traces in a differentially private manner~\cite{PP-ProcessM}, query-based information extraction only works for a limited number of queries to guarantee privacy. With every query, additional information about underlying sensitive traces is leaked. As an alternative, if the mining strategy itself guarantees differential privacy, the resulting process representation could be arbitrarily used, e.g., for generating synthetic event traces, without causing any additional privacy leakage. 

    This paper introduces a novel privacy-preserving process mining algorithm called \texttt{Differentially Private Inductive Miner} (DPIM) which produces a process tree (PST) based on an event log. DPIM replaces privacy-leaking operations of the Inductive Miner on single traces with privacy-compliant operations on sets of traces. These operations are designed such that we show strong differential privacy guarantees while approximating the functionality of the Inductive Miner as shown by the common quality measures (i.e., fitness, precision, simplicity, and generalization). We evaluated our algorithm against the Inductive Miner based on 14 real-world event logs in terms of accuracy and privacy.\footnote{Code, data, and evaluations are available at \url{github.com/Schulze-M/DPIM}} The trade-off between privacy gain and data utility loss depends on the chosen degree of $\varepsilon$ (lower means more privacy), event log complexity (simpler is better), and event log size (larger is better). In Fig.~\ref{fig:eval_fitness} we quantify this trade-off and show that differential privacy can be obtained on real-world event logs with process models that keep a fitness of 0.95, precision of 0.9, simplicity of 0.7, and generalization of 0.8.

    \begin{figure*}[!ht]
    \centering
    \begin{minipage}[b]{0.34\textwidth}
        \includegraphics[width=\linewidth]{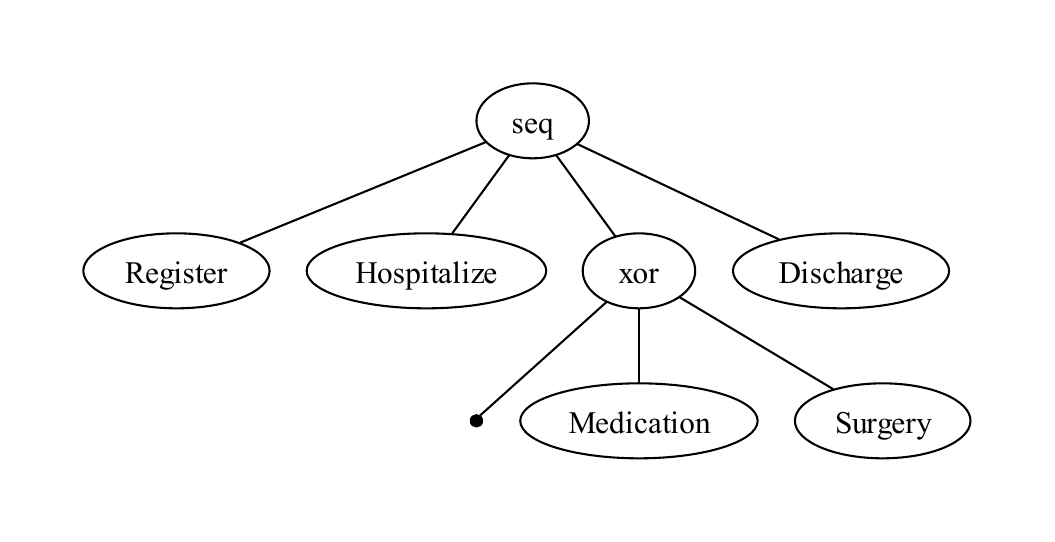}
        \caption{PST on Trace Variants 1 to \emph{3}}
        \label{fig:medical_process}
    \end{minipage}
    \begin{minipage}[b]{0.34\textwidth}
        \includegraphics[width=\linewidth]{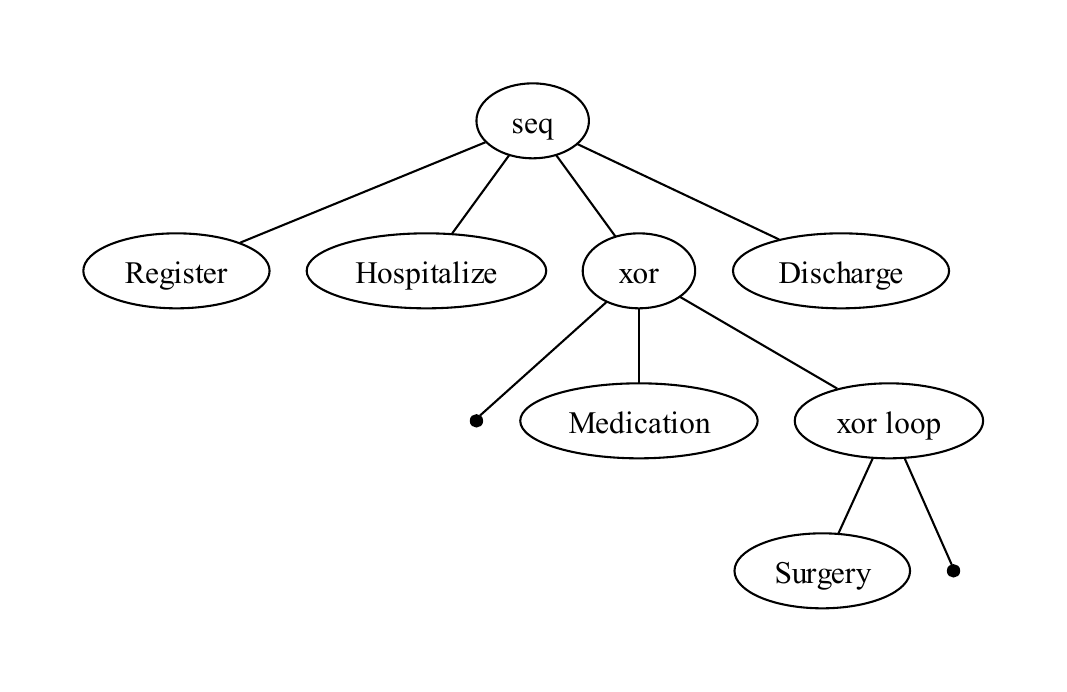}
        \caption{PST on Variants 1 to \emph{4}}
        \label{fig:patient_jane}
    \end{minipage}
    \begin{minipage}[b]{0.3\textwidth}
    \makeatletter
    \renewcommand{\@captype}{table}
    \makeatother
        \begin{tabular}{c|c|l}
        \toprule
            Variant & Count & Trace \\
            \midrule
             \emph{1} & 63x & $\langle$\,R,\,H,\,M,\,D\,$\rangle$ \\
             \emph{2} & 25x & $\langle$\,R,\,H,\,S,\,D\,$\rangle$ \\
             \emph{3} & 12x & $\langle$\,R,\,H,\,D\,$\rangle$ \\
             \hdashline
             \emph{4} & 1x & $\langle$\,R,\,H,\,S,\,S,\,D\,$\rangle$ \\
             \emph{5} & 0x & $\langle$\,R,\,M,\,D\,$\rangle$ \\
             \bottomrule
        \end{tabular}
        \caption{Exemplary Trace Log}
        \label{tab:probstat}
    \end{minipage}
    \end{figure*}

    \pparagraph{Structure.}
    Sec.~\ref{Sec:Problem_Statement} defines the problem on an exemplary use case. Sec.~\ref{Sec:Related_Work} compares our approach to related approaches. Sec.~\ref{Sec:Preliminaries} presents how to incorporate privacy guarantees into the Inductive Miner. Sec.~\ref{Sec:Approach} describes our \emph{DPIM}. Evaluation results are summarized in Sec.~\ref{Sec:Evaluation}, and a conclusion is drawn in Sec.~\ref{Sec:Conclusion}.

\section{Problem Statement}
    \label{Sec:Problem_Statement}
    Process mining algorithms extract insights about an organization's processes from recorded and potentially sensitive event log data. Yet, it carries the inherent risk that what is disclosed may be private. We assume a medical treatment process in a hospital (Fig. \ref{fig:medical_process}) where patients are \emph{registered}, \emph{hospitalized}, and in slight medical cases can be \emph{discharged} immediately, while in more severe medical cases they either receive \emph{medication} or undergo \emph{surgery}, before being \emph{discharged}. Each trace represents a possibly unique path in the treatment process that can be directly linked to the health situation of an individual patient. Thus, anonymizing an event log can not keep the original traces unmodified as the traces pose a re-identification risk otherwise \cite{ReIDEvent}.

    On a process model based on unmodified traces, we can construct the following attack: assume an attacker, e.g. a data controller or business analyst, with authorized access to query aggregated data such as process models, but no direct access to the underlying data. The attacker has the background knowledge to perform a \emph{difference attack} \cite{foundation}, e.g. the attacker queries the process model at a timepoint where the trace variants \emph{1} to \emph{3} in Table \ref{tab:probstat} are included (\textit{R:Register, H:Hospitalize, M:Medication, S:Surgery, D:Discharge}) (cf. Fig. \ref{fig:medical_process}). After learning that Jane visited the hospital, the attacker queries again and finds out that the process model changed as seen in Fig. \ref{fig:patient_jane} and deducts that a trace variant like number \emph{4} in Table \ref{tab:probstat} caused that difference. Thus, the attacker learns that Jane must have undergone surgery at least twice.
    
    To counter privacy attacks without significantly compromising utility, we seek to anonymize traces by summarizing as many common behaviors as possible, e.g. by working on the trace variants or directly-follows relation. Intuitively, the more persons exhibit a common behavior, like having the same trace or directly following activities, the better a single person can blend in the crowd. Thus, the more we summarize, the stronger the privacy becomes. We can accelerate this advantage with strong privacy protection mechanisms like differential privacy (DP), where a high-utility DP mechanism has a privacy protection level $\eps \in \mathcal{O}(\sfrac{1}{n})$ where $n$ is the number of summarized elements.

    Based on summarized traces, we propose to construct a process structure tree (PST) where it is possible to synthesize an event log or build a Petri net. To create a privacy-preserving PST, we propose a method called \emph{DPIM} that is based on the Inductive Miner and uses DP. DP guarantees that the resulting PST is protected against the addition, removal, or change triggered by any person's trace. Thus, even those attackers with unlimited background knowledge are unable to infer the influence of a single trace on the process model and thereby recreate the nature of the trace.

\section{Related Work}
    \label{Sec:Related_Work}
    Closely related to our approach are approaches that focus on differentially private event log sanitization \cite{elkoumy_differentially_2022, Fahrenkrog_Semantics-Aware,PP-ProcessM, DBLP:conf/bpm/Fahrenkrog-Petersen20,PrivPre-DFG}. Yet, two of them \cite{elkoumy_differentially_2022,PrivPre-DFG} admit to privacy limitations which leaves the privacy protection unclear; hence we do not compare our utility to their approach. For the other two papers \cite{Fahrenkrog_Semantics-Aware,DBLP:conf/bpm/Fahrenkrog-Petersen20}, we found a counterexample of a component of their approaches in the form of a privacy attack that violates the differential privacy property. Hence, it is unclear which degree of privacy protection these papers achieve, rendering a direct utility comparison unfair.
    
    Elkoumy et al. \cite{PrivPre-DFG} aims to generate a directly-follows graph (DFG) differentially private by determining the amount of noise needed and then noising the weight of the arcs. However, they state that they do not add or delete arcs. 
    Yet, as per our example in Table.~\ref{tab:probstat} the person with trace variant \emph{4} can alter whether the DFG contains the activity pair $\langle\,S,\,S\,\rangle$. Thus this person can add or remove an arc in a DFG since any activity pair is represented in a DFG. Hence, their protection method is not DP, as an attacker can observe in the DFG whether \emph{4} was present.

    Elkoumy et al. \cite{elkoumy_differentially_2022} anonymizes event log timestamps and removes some traces, assuming limited background knowledge of an attacker. This removal is based on prior knowledge. However, they state that no new trace variants are introduced. Yet, as per our example in Table~\ref{tab:probstat} and Fig.~\ref{fig:patient_jane}, the person with trace variant \emph{5} can alter the process by introducing a new unforeseen trace that is not present in the data, e.g. the activity pair $\langle R, M \rangle$. This indicates that their method does not meet differential privacy standards, as an attacker with extensive background knowledge can detect the presence of trace variant \emph{5} in the anonymized data.

    Fahrenkrog et al. \cite{Fahrenkrog_Semantics-Aware} propose a DP algorithm that releases an anonymized trace variant distribution using the exponential mechanism in combination with Laplace noise. From the trace variant distribution, their algorithm builds a DFR. Although they add noise to selected activity pair frequencies of the DFR, the algorithm preselects a limited number of activity pairs based on the event log, which is data-dependent. This preselection, driven by semantic correctness or a k-follows score, undermines differential privacy. As new trace variants can add a huge number of new plausible activity pairs which put previously distant activities closer and leads to a lower k-follows score, a huge number of new activity pairs are accepted. Thus, an attacker can infer whether the new trace variant is used by how many activity pairs are released.

    The PRIPEL framework\cite{DBLP:conf/bpm/Fahrenkrog-Petersen20} anoynmizes traces using a random timestamp shift and noise added to the count of trace variants. The usefulness of these anonymized traces is improved by selecting only those similar to the original traces via the Levensthein distance. Yet, as per our example in Table~\ref{tab:probstat}, the person with trace variant \emph{4} could influence this selection if its trace is similar to an anonymized one. This suggests that the method does not ensure differential privacy, as an attacker can infer the presence of trace variant \emph{4} from the selection process.
    
    Mannhardt et al. \cite{PP-ProcessM} is similar to our approach as it introduces a method to generate the frequency of all directly-follows relation (DFR) or a prefix-tree differentially private. Our DPIM concentrates on generating a faithful process tree (PST), which includes a parallel and loop cut selection.

    Further related to our work are research areas that focus on privacy measures that do not provide an equally strong guarantee of privacy as DP or are parallel to our work by protecting the mining process cryptographically. K-anonymity methods to protect event logs have been proposed \cite{rafiei_group-based_2021, DBLP:conf/rcis/RafieiWA20}. However, these are not robust against unlimited background knowledge like differential privacy (DP) and are hence vulnerable to intersection attacks \cite{PSO}. Rafiei et al. \cite{rafiei_towards_2021, rafiei_quantifying_2022} suggests disclosure risk quantification measures. These do not provide techniques to guarantee privacy.     Burattin et al. \cite{AnonPM} suggest outsourcing process mining while maintaining the confidentiality of event logs and discovered process models using encryption and cryptography, yet the decrypted model itself is not protected with DP.

\section{Preliminaries}
    \label{Sec:Preliminaries}
    \noindent\textbf{Differential Privacy (DP) \cite{DP}}
        \label{Sec:DifPriv}
        is a mathematical framework that enables the analysis and sharing of sensitive data while preserving the privacy of individual records within the data set. It guarantees privacy by introducing a controlled amount of noise to the data, ensuring that the presence or absence of any individual's data does not significantly influence the results of queries. This privacy guarantee is quantified by a privacy parameter $\epsilon$, with lower values offering stronger privacy protection at the expense of reduced data utility. The \emph{post-processing theorem} \cite{privacybook} states that no further processing of the output of a DP mechanism can increase privacy leakage.

        \begin{definition}[\textbf{Bounded Differential Privacy}]
            \label{def:dp}
            Two event logs $L, L'$ are neighboring (written $L\sim L'$) if they differ in at most one trace. Let $\mathcal{R}$ be the set of random variables over some set $O$.
            A randomized mechanism $\mathcal{M}: \mathcal{L} \rightarrow \mathcal{R}$ is $\epsilon$-differentially private, with $\epsilon > 0$, if for all $\mathcal{S} \subseteq O$, for all $n$, and all neighboring event logs $L$, $L'$ $\in$ $\mathcal{L}$ with $ n = |L| = |L'|$: $\Pr[{M}(L) \in \mathcal{S}] \leq \exp(\epsilon)\Pr[\mathcal{M}(L') \in {S}]$.
        \end{definition}
    
    \noindent\textbf{The Laplace Mechanism}
        \label{Sec:Laplace}
        $\mathcal{M}_{L,q,\eps}$ \cite{DP} is $\eps$-DP and takes a database $D$ as input and outputs a noised $\Delta_q$-sensitivity bounded query $q$: $\mathcal{M}_{L,q,\eps}(D) \mapsto q(D) + \mathrm{Lap}(0, \frac{\Delta_q}{\epsilon})$. The sensitivity $\Delta_q$ describes how much the output could change in the worst-case if a single data point is exchanged and $\mathrm{Lap}(\Delta_q/\epsilon)$ is defined as
        $\mathrm{Lap}(\Delta_q/\eps)[o] \coloneqq \frac{\eps}{2\Delta_q} \exp(-|o - \mu|\eps/\Delta_q)$. Thus, a smaller $\eps$ or a larger sensitivity $\Delta_q$ corresponds to noise with a larger standard deviation.
        
        \begin{theorem}[The Laplace Mechanism is DP] \label{def:Laplace}
            For a $\Delta_q$-bounded ($\Delta_q \in \mathbb{R}_+$) counting query $q$ and an $\eps > 0$, the Laplace Mechanism $\mathcal{M}_{L,q,\eps}$ is $\eps$-DP.
        \end{theorem}
    
    \noindent\textbf{Report Noisy Max (ReNoM) \cite{privacybook}}
        \label{Sec:NoisyMax}
        is an $\eps$-DP algorithm that takes a set of $m$ counting queries and returns the index of the query with the highest noisy count. For that, we add independently sampled noise $\mathrm{Lap}(1/\epsilon)$ to each query (e.g. the frequency of activity pairs in a directly-follows relation).
    
    \noindent\textbf{Rejection Sampling (RejSamp)} \label{Sec:RejectionSampling}
        \begin{algorithm}[t]
        \caption{DP Rejection Sampling \cite[Algorithm 1]{privateCanidateSelection}} 
        \label{algo:baseReject}
            \begin{algorithmic}
                \State {\bfseries Input:} threshold $t$, probability $\gamma \leq 1$, privacy budget $\epsilon_0 \leq 1$, number of steps $T \geq \max{\left(\frac{1}{\gamma}\ln \frac{2}{\epsilon_0}, 1 + \frac{1}{e\gamma}\right)}$, $\epsilon_1$-DP mechanism $M(D)$
                \For{$j = 1,\dots,T$}
                    \State draw $(x, q) \sim M(D)$
                    \If{$q \geq t$}
                        \Return $(x, q)$
                    \EndIf
                    \State flip $\gamma$-biased coin s.t. with probability $\gamma$: \Return $\bot$
                \EndFor
                \State \Return $\bot$
            \end{algorithmic}
        \end{algorithm}
        as in Alg.~\ref{algo:baseReject} is a technique to privately select candidates from a mechanism $M(D)$ that outputs $x$ (e.g. a PST) and a score $q$ (e.g. a noisy fitness). It is $\epsilon$-DP with $\eps = 2 \cdot \epsilon_1 + \epsilon_0$ if $M(D)$ is $\eps_1$-DP \cite[Algorithm 1]{privateCanidateSelection}. With a pre-defined threshold $t$, we accept and return $(x,q)$ if $q \ge t$, otherwise we repeat the process at most $T$ rounds until either the output is accepted or a $\gamma$ biased coin is positive.

\section{Approach}
    \label{Sec:Approach}
    Our novel privacy-preserving process miner \emph{DPIM} produces a process tree (PST) based on an event log. A key ingredient of many process mining algorithms is the directly-follows relation, which describes which pair of activities directly follow each other in the traces. DPIM (cf. Algorithm~\ref{algo:reject}) leverages two key insights: it suffices to build a representative PST (cf. Algorithms~\ref{algo:buildT} and \ref{algo:appendTree}) to solely operate on $(i)$ a directly-follows relation where each activity pair is annotated with its frequency in the event log and $(ii)$ the first and last activities of any loop or parallel cut (cf. Algorithm~\ref{algo:start_end}). In Section~\ref{sec:privacyproof}, we prove that DPIM is differentially private.

    \begin{algorithm}[ht]
    \caption{DPIM: Creates a differentially private PST.} \label{algo:reject}
    \begin{algorithmic}[1]
        \State {\bfseries Input:} EventLog L, fitness threshold $t$, probability $\gamma$, privacy budgets $\epsilon, \epsilon_0$, budget share $(r_1, r_2, r_3)$, lower and upper bound $lb,ub$
        \item[]\vspace{-0.5em} 
            \State $\epsilon_1 \gets 0.5\cdot(\epsilon - \epsilon_0)$ \algorithmiccomment{Budget build PST \& fitness}
            \State $T \geq \max\{\,\frac{1}{\gamma}\ln\frac{2}{\epsilon_0}, 1+\frac{1}{e\gamma}\,\}$ \algorithmiccomment{\#iterations for RejSamp}
            \LeftComment{Initialize DFR dict incl. dummy start \& end activities}
            \State DFR \parbox[t]{6.5cm}{$\gets \{ \{ (a,b)\colon 0 \} \mid (a,b) \in (A \cup \{\textsc{start}\}) \times (A \cup \{ \textsc{end} \})\}$}
            \item[]
            \LeftComment{Count\,activity\,pairs\,in\,DFR:\,in\,how\,many\,traces\,it\,occurs}
            \For{dfr \textbf{in} DFR}
                \For{trace \textbf{in} L}
                    \If{dfr $\in$ trace}
                        dfr.value $\gets \text{dfr.value} + 1$
                    \EndIf
                \EndFor
            \EndFor
            \item[]\vspace{-0.5em}
            \LeftComment{RejSamp until the PST meets the fitness threshold}
            \For{$1, \dots, T$} \algorithmiccomment{\dots for at most $T$ rounds}
                \State n $\gets \mathrm{Uniform}(\textit{lb}, \textit{ub})$
                \LeftComment{Chose $n$ DFR with the highest noisy count (ReNoM)}
                \State DP-DFR $\gets \{\,\,\}$
                \For{$1,\dots, \text{n}$}
                    \State dfr $\gets \text{DFR}[\textsc{ReNoM}(\text{DFR}~\setminus~\text{DP-DFR}, \frac{\eps_1 \cdot r_1}{2\cdot n})]$
                    \State dfr.value $\gets \text{dfr.value} + \mathrm{Laplace}(0, \frac{2\cdot\text{n}}{\epsilon_1 \cdot r_1})$
                    \State DP-DFR $\gets \text{DP-DFR} \cup \{\, \text{dfr} \,\}$
                \EndFor
                \State DP-PST $\gets$ \textsc{buildT}(DP-DFR, L, $\sqrt{8}\cdot\frac{2\cdot\text{n}}{\epsilon_1 \cdot r_1}$, $\epsilon_1 \cdot r_2$, DP-S =$\emptyset$, DP-E=$\emptyset$)\algorithmiccomment{Algorithm~\ref{algo:buildT}}
                \If{ $\text{DP-fit} \geq t$}
                    \Return{(DP-PST, DP-fit)}
                \EndIf
                \State flip a $\gamma$-biased coin s.t. with prob. $\gamma$: \Return{$\bot$}
            \EndFor
            \State \Return{$\bot$}
    \end{algorithmic}
    \end{algorithm}

    \subsection{High-level Description of DPIM}
    \begin{wraptable}{r}{0.51\columnwidth}
        \centering
        \begin{tabular}{l|c|c}
             \toprule
             & \multicolumn{2}{c}{Count}\\
             \cmidrule(l){2-3}
             DFR & Raw & Noisy \\
             \midrule
             (\textsc{start}, R) & 100 & 105.69\\
             (R, H) & 100 & 97.23 \\
             (R, \textsc{end}) & 0 & 5.99 \\
             $\cdots$ & $\cdots$  & $\cdots$\\
             (H, S) & 25 & 22.31 \\
             (S, S) & 0 & 7.64 \\
             (S, D) & 25 & 31.02 \\
             (\textsc{start}, S) & 0 & -2.16\\
             \bottomrule
        \end{tabular}
        \caption{DFR based on trace variants \emph{1} to \emph{3} of Table \ref{tab:probstat}}
        \label{tab: dfr_sample}
    \end{wraptable}
    DPIM (cf. Algorithm \ref{algo:reject}) generates a process tree (PST) and -- for the rejection sampling step -- a fitness score in a differentially private manner. We use rejection sampling (cf. Alg.~\ref{algo:baseReject}) to accept only those PSTs with a high fitness. DPIM works as follows: it annotates in the directly-follows relation \texttt{DFR} in how many traces of the event log $L$ any two activities directly follow each other (Alg. 2, lines 5-7). Note that for privacy reasons, we only use a binary count of whether two activities directly follow. We consider any permutations of activity pairs, including dummy start and end activities (line 4).
    Table \ref{tab: dfr_sample} provides sample DFRs with counts based on Table \ref{tab:probstat}. As this DFR contains sensitive data (counts based on the number of traces), we modify this relation in the next steps to make it differentially private. We perform at most $T$ rounds of rejection sampling, which rejects a generated PST (cf. Algorithm~\ref{algo:buildT}) that does not have a fitness of at least $t$.
    Noise induced by differential privacy increases the variance of the PSTs, which can result in PSTs with low fitness. As the fitness depends on the sensitive event log, we noise the fitness with Laplace noise (cf. Theorem~\ref{def:Laplace}). The scale parameter of the Laplace noise that is added to the fitness, i.e. $\sfrac{1}{|L|\cdot\eps_1\cdot r_3}$, directly correlates with the expected deviation on how much the fitness threshold is missed. We design DPIM in a way that the PST generation in \textsc{buildT} (cf. Algorithm~\ref{algo:buildT}) operates on the annotated directly-follows relation (\texttt{DP-DFR}) that is obtained differentially private. The post-processing theorem of differential privacy \cite{privacybook} states that any operation on this relation does not incur additional privacy leakage. We select the top-$n$ frequency-annotated activity pair in \texttt{DP-DFR} in a differentially private way: in each iteration of rejection sampling, the algorithm first samples the variable $n$ uniformly at random from lower- and upper-bound hyperparameters $\text{lb},\text{ub}$ (lines 8-9). Directly choosing $n$ as the number of activity pairs that have a frequency count larger than $0$ would have privacy leakage.
    Due to rejection sampling, the effect of a suboptimal $n$ selection is limited to how much including non-existent or excluding existing activity pairs affects the fitness. Next, we select the index of the directly-follows relation with the largest noisy count using the differentially private Report Noisy Max mechanism (cf. Section~\ref{Sec:NoisyMax}) and also release the noisy count which we formally noise again. We repeat this \emph{argmax} process $n$ times where we exclude the previously selected indices in each round. As a result, we obtain $n$-many frequency-annotated activity pairs in \texttt{DP-DFR} that occur probably the most in the event log (lines 11-14). For example, when having $n=5$, we select (\textsc{start}, R), (R, H), (H, S), (S, D) and (S, S), which have highest noisy count in Table \ref{tab: dfr_sample}, even though \emph{(S, S)} is not part of the first three trace variants. These selected DFRs are then used to build the PST (line 15) (cf. Algorithm~\ref{algo:buildT}).
    
    \subsection{Detailed Description of the Subalgorithms}    
    \noindent\textbf{buildT (Algorithm~\ref{algo:buildT})}
    \begin{algorithm}[t]
    \caption{BuildT: Builds a PST by detecting cuts in a DFR.}
    \label{algo:buildT}
        \begin{algorithmic}[1]
            \State {\bfseries Input:} DP-DFR, args=\{L, std, $\epsilon_{\text{start\_end}}$, DP-S, DP-E\}
            \If{first recursive round}
                \LeftComment{calculate start \& end activities and event log size DP}
                \State DP-S $\gets \{b \mid a = \textsc{start}\}~\forall \{(a,b): p\} \in \text{DP-DFR}$
                \State DP-E $\gets \{a \mid b = \textsc{end}\}~\forall \{(a,b): p\} \in \text{DP-DFR}$
                \State DP-ESize $\gets \sum_{\{(a,b): p\} \in \text{DP-DFR} } p \cdot 1[a = \textsc{start}]$
                \State DP-ActC $\gets \{a\colon 0 \mid a \neq \textsc{start} \land b \neq \textsc{end}\}~\forall\{(a,b): p\} \in \text{DP-DFR}$
                \For{$\{(a,b): p\}$ \textbf{in} DP-DFR} \algorithmiccomment{activity count}
                    \State DP-ActC[a] $\gets$ DP-ActC[a] + p
                \EndFor
            \EndIf
            \item[]\vspace{-0.5em}
            \State seqSet $\gets$ \textsc{sequence}(DP-DFR) \algorithmiccomment{cf. \cite{InductiveMiner}}
            \If{len(seqSet) $>$ 1} 
                \Return \textsc{appendTree}($\rightarrow$, seqSet, DP-DFR, args, DP-ESize, DP-ActC) \algorithmiccomment{seq. cut, Alg.~\ref{algo:appendTree}}
            \EndIf
            \item[]\vspace{-0.5em}
            \State xorSet $\gets$ \textsc{xor}(DP-DFR) \algorithmiccomment{cf. \cite{InductiveMiner}}
            \If{len(xorSet) $>$ 1} 
                \Return \textsc{appendTree}($\bigotimes$, xorSet, DP-DFR, args, DP-ESize, DP-ActC) \algorithmiccomment{xor cut}
            \EndIf
            \item[]\vspace{-0.5em}
            \State DP-DFR' = DP-DFR \algorithmiccomment{remove loops}
            \For{$\{(a,b): p\}$ and $\{(b,a): p'\}$ \textbf{in} DP-DFR}
                \If{$p + p'$ $\geq$ DP-ESize + std}
                    \State DP-DFR' $\gets$ delete $\{(a,b): p\}, \{(b,a): p'\}$ 
                \EndIf
            \EndFor
            \item[]\vspace{-0.5em}
            \State andSet $\gets$ \textsc{and}(DP-DFR', DP-S, DP-E) \algorithmiccomment{cf. \cite{InductiveMiner}}
            \If{len(andSet) $>$ 1} 
                \Return \textsc{appendTree}($\bigoplus$, andSet, DP-DFR', args, DP-ESize, DP-ActC) \algorithmiccomment{parallel cut}
            \EndIf
            \item[]\vspace{-0.5em}
            \State loopSet, DP-DFR-nL $\gets$ \textsc{loop}(DP-DFR, DP-S, DP-E)
            \If{len(loopSet) $>$ 1}
                \Return \textsc{appendTree}($\circlearrowleft$, loopSet, DP-DFR-nL, args, DP-ESize, DP-ActC) \algorithmiccomment{loop cut}
            \EndIf
            \item[]\vspace{-0.5em}
            \State \parbox[t]{8cm}{\Return \textsc{appendTree}(FLOWER, s, DP-DFR, args, DP-ESize, DP-ActC)}
        \end{algorithmic}
    \end{algorithm}
    \textsc{buildT} analyses the differentially privately obtained DFR (\texttt{DP-DFR}) and returns a process tree (PST) by determining sequential ($\rightarrow$), exclusive or ($\bigotimes$), parallel ($\bigoplus$), and loop ($\circlearrowleft$) cuts in \texttt{DP-DFR}. In this step, we follow the Inductive Miner \cite{InductiveMiner}. %as these cut operations mostly do not operate on sensitive data. 
    \textsc{buildT} works as follows:

    Initially, the algorithm determines the start and end activities (\texttt{DP-S} and \texttt{DP-E}), counts the initial event log size (\texttt{DP-ESize}), and counts activity occurrences at the start of a DFR (\texttt{DP-ActC}) (lines 3-8). These operate on \texttt{DP-DFR} and do not incur additional privacy leakage due to the post-processing theorem.

    \emph{DP-S, DP-E.} The start and end counts are necessary for loop detection and to distinguish parallel from loop cuts (lines 3-4).
    
    \emph{DP-ESize.} We count the event log size to exclude simple loops for an improved parallel cut detection and to determine silent transitions $\tau$. We use the design of \texttt{DP-DFR} and sum up all counts in \texttt{DP-DFR} that contain the dummy \textsc{start} activity (line 5). The original Inductive Miner uses a different technique, which performs privacy-leaking lookups on the sensitive event log. Our technique is differentially private but constitutes an overapproximation, which results in a higher amount of $\tau$ in the resulting PST. The reason is that \texttt{DP-ESize} does not change in a subtree, although the active event log size in this subtree does if, e.g., this subtree is branched by an XOR.
    
    \emph{DP-ActC.} We count how often each activity occurs at the start of a DFR which determines with \texttt{DP-ESize} whether a subtree is optional, i.e. XOR(subtree, $\tau$). We sum over those counts of \texttt{DP-DFR} where each activity is at the beginning of the relation, excluding activity pairs that involve the dummy activities \textsc{start} or \textsc{end} (lines 6-8).
    
    \emph{Approximating IM.} Next, we detect the cuts using the \texttt{DP-DFR} for the \textsc{appendTree} subroutine (cf. Algorithm~\ref{algo:appendTree}). For SEQ and XOR cuts, we use the Inductive Miner (IM) \cite{InductiveMiner} as these operate on a DFR and not on the event log (lines 9-12). For AND and LOOP cuts, we adapt the IM by using only the \texttt{DP-DFR} and the knowledge of the start and end activities within a LOOP and AND candidate, \texttt{DP-S} and \texttt{DP-E}. To simplify the parallel cut detection, the DPIM removes simple loops like LOOP(A,B) or LOOP(AND(A,B,C), $\tau$) over activities A, B, and C a prior (lines 13-16). For instance, the directly-follow relation of LOOP(A,B) is similar to that of the parallel cut AND(A,B): $(A,B)$ and $(B,A)$. However, the frequency count differs: for a loop, $\text{count}(A,B) + \text{count}(B,A) > \text{EventSize}$, whereas for a parallel cut, $\text{count}(A,B) + \text{count}(B,A) = \text{EventSize}$. As the counts and event log size are noisy, we only remove loops that are considerably above the noisy event log size \texttt{DP-ESize} (factor: $\sqrt{8}\cdot \text{LaplaceScale}$) (line 15-16).
    
    \pparagraph{appendTree (Algorithm \ref{algo:appendTree})}
    \begin{algorithm}[t]
    \caption{appendTree: Recursively appends a cut to a PST.}
    \label{algo:appendTree}
    \begin{algorithmic}[1]
        \State {\bfseries Input: }{cutType c, cutSet s, DP-DFR, p=\{L, std, $\eps$, DP-S, DP-E\}, DP-ESize, DP-ActC}
        \item[]\vspace{-0.5em}
            \If{c = LOOP \textbf{or} c = AND}
                $\eps$ $\gets$ $0.5\cdot\eps$
            \EndIf
            \State DP-S, DP-E $\gets$ \textsc{detectS\_E}(c, s, DP-DFR, p)
            \algorithmiccomment{Alg.~\ref{algo:start_end}}
            \If{c = LOOP}
                subtree[$\bot$] $\gets$ XOR($\tau$, $\bot$) \algorithmiccomment{Loops are always optional}
            \EndIf
            \State subtree[$\bot$] $\gets$ c[\dots,$\bot$,\dots]\algorithmiccomment{Add cut to subtree}
            \item[]\vspace{-0.5em}
            \For{acts \textbf{in} s}
                \If{len(act) = 1}
                \algorithmiccomment{create leaf node DP}
                    \If{\{(acts[0], acts[0]): p\} $\in$ DP-DFR}
                        \State subtree[$\bot$] $\gets$ LOOP(acts[0], $\tau$)
                    \ElsIf{DP-ActC[acts[0]] $\leq$ DP-ESize - std}
                        \State subtree[$\bot$] $\gets$ XOR($\tau$, acts[0])
                    \Else{}
                        subtree[$\bot$] $\gets$ acts[0]
                    \EndIf
                \ElsIf{len(act) = 0}
                    \State subtree[$\bot$] $\gets$ $\tau$
                \Else\algorithmiccomment{Alg.~\ref{algo:buildT} to cut remaining activities recusively}
                    \State subtree[$\bot$] $\gets$ \textsc{buildT}(DP-DFR[acts], p)
                    \State \Return subtree
                \EndIf
            \EndFor

            \If{$\sum_{\text{acts} \in s}\sum_{a \in \text{acts}} \text{DP-ActC}[a] \leq$ DP-ESize - std}
                 \State subtree[$\bot$] $\gets$ $\tau$
            \EndIf
            \State \Return subtree
    \end{algorithmic}
    \end{algorithm}
    \textsc{appendTree} performs the cut of Algorithm~\ref{algo:buildT} and either returns a leaf in a PST if the cut separates a single activity (line 20) or a recursively built subtree on the remaining activities in a cut (line 16) (cf. \textsc{buildT} of Algorithm~\ref{algo:buildT}).

    \textsc{detectS\_E} spends some privacy budget $\eps$ to recalculate the start and end activities within a loop or parallel cut for a subsequent recursive round (lines 2-3). A priori the count of parallel or loop cuts is unknown. Therefore, we spend the privacy budget using a geometric series, where we first spend $0.5\eps$, then $0.25\eps$, etc (line 2).
    
    The algorithm builds the PST recursively where we append cuts and activities to an empty leaf $\bot$ with the notation ``subtree[$\bot$] $\gets$''. In particular, ``subtree[$\bot$] $\gets$ c[\dots,$\bot$,\dots]'' denotes that we add the cut c with as many empty leaves as needed (line 5), and ``subtree[$\bot$] $\gets$ c[acts[0]]'' marks a leaf with one activity where no subsequent cuts can be appended (i.e, line 12). We overapproximate loops by making them optional (line 4).

    After appending the cut, the algorithm appends for each element in the cut set either a leaf or a subtree (lines 6-17). For instance, for input \texttt{cutSet=[[R], [H], [S,D]]} and \texttt{cutType=SEQ} on activities R, H, S, and D, we append a sequence cut with two leaves with the R and H and a subtree over S and D. For appending each subtree we call Algorithm~\ref{algo:buildT} recursively on the directly-follows relation \texttt{DP-DFR[acts]} relevant for the subtree (line 16). Note that the algorithm also handles a few special cases: first, a self-loop where we loop over one activity or $\tau$ (line 9); second, an optional activity represented by XOR($\tau$, activity) which only happens if this activity occurs considerably (factor: $\sqrt{8}\cdot \text{LaplaceScale}$) less then \texttt{DP-ESize} (line 11); and third if all leaves of this cut are optional as each activity within it occurs considerably less than \texttt{DP-ESize} (lines 18-19). The algorithm also overapproximates $\tau$-transitions in the second and third special cases, as it considers the overall event log size \texttt{DP-ESize} and not the active number of traces within this cut.

    \pparagraph{detectS\_E (Algorithm \ref{algo:start_end})}
    \begin{algorithm}[t]
    \caption{DetectS\_E: Detects the next start\,\&\,end activities.}
    \label{algo:start_end}
    \begin{algorithmic}[1]
        \State {\bfseries Input:}cutType c, cutSet s, DP-DFR, p=\{DP-S, DP-E, L, \_, $\eps$\}
            \If{c = $\circlearrowleft$ \textbf{or} c = $\bigoplus$}
                \For{trace \textbf{in} L} \algorithmiccomment{only keep event log of subtree}
                    \State trace $\gets$ del. all activities in trace that are not in s
                \EndFor

                \State cStart $\gets \{\, \{\,a: 0\,\} ~\forall a \in \text{acts}, \text{acts} \in s \,\}$ 
                \State cEnd $\gets \{\, \{\,a: 0\,\} ~\forall a \in \text{acts}, \text{acts} \in s \,\}$
                \item[]\vspace{-0.5em}
                \LeftComment{count for each activity how often it starts and ends}
                \For{trace \textbf{in} L}
                    \State cStart[trace.firstAct] $\gets$ cStart[trace.firstAct] + 1
                    \State cEnd[trace.lastAct] $\gets$ cEnd[trace.lastAct]  + 1
                \EndFor
                
                \For{\{\,a : cnt\,\} \textbf{in} cStart \textbf{and} \{\,a : cnt\,\} \textbf{in} cEnd}
                    \State cnt $\gets$ cnt + $\mathrm{Laplace}(0, \frac{4}{\epsilon})$ \algorithmiccomment{noise this count}
                \EndFor
                \State DP-S $\gets \{\, \{\,a: \text{cnt}\,\} \in \text{cStart} \mid \text{cnt} \ge \sqrt{8}\cdot\frac{4}{\eps} \,\}$\algorithmiccomment{keep significant count}
                \State DP-E $\gets \{\, \{\,a: \text{cnt}\,\} \in \text{cEnd} \mid \text{cnt} \ge \sqrt{8}\cdot\frac{4}{\eps} \,\}$
                \State \Return DP-S, DP-E
            \item[]\vspace{-0.5em}
            \ElsIf{c = $\bigotimes$ or c = $\rightarrow$}
                \For{acts \textbf{in} s}
                    \If{len(acts) = 1}
                        \State DP-S $\gets$ replace acts[0] with DP-DFR-succ(acts[0])
                        \State DP-E $\gets$ replace acts[0] with DP-DFR-pred(acts[0])
                    \EndIf
                \EndFor
                \State \Return DP-S, DP-E
            \EndIf
            
    \end{algorithmic}
    \end{algorithm}
    \textsc{detectS\_E} determines the start and end activities within the current cut. For a sequence ($\rightarrow$) or exclusive or ($\bigotimes$) cut, we determine these by looking at the predecessor or successor of the start or end activities of the previous cut in \texttt{DP-DFR} (lines 15-20). This step is DP by the post-processing theorem \cite{privacybook}. For a parallel ($\bigoplus$) or loop cut ($\circlearrowleft$), we work on the event log where we spend a privacy budget roughly proportional to the number of ANDs and LOOPs (lines 2-14). This is due to the nature of ANDs and LOOPs where a successor or predecessor is not necessarily part of the next cut.
    
    \emph{AND or LOOP case}. First, we identify the correct start and end activities by deleting all activities in the event log $L$ that are not part of the current cutSet (lines 3-4), e.g. assuming (S, S) and (S, D), the traces would be: $\langle\,D\,\rangle$ and $\langle\,S,\,D\,\rangle$, based on trace variants \emph{1} to \emph{3}. Thus, we only have to count in \texttt{cStart} and \texttt{cEnd} the first and last activity of each modified trace (lines 7-9). Second, we noise these counts using the Laplace Mechanism (lines 10-11) (cf. Section~\ref{sec:privacyproof}). Each trace has a single start and end activity, so it can only influence one count per start and end. Thus, we apply parallel composition to prevent scaling the noise with the number of candidates. Third, we select those activities from the candidate sets \texttt{cStart} and \texttt{cEnd} as start and end activities \texttt{DP-S} and \texttt{DP-E}, which have a count significantly greater than zero (lines 12-13). As a significance level for this particular case, we use the doubled standard deviation $2\sigma = \sqrt{8}\cdot b$ of the added Laplace noise $\mathrm{Laplace}(0,b)$, which corresponds to a $97\,\%$ significance level of sampling below the standard deviation.

    \subsection{Privacy Guarantee}\label{sec:privacyproof}
    \begin{theorem} \label{dp_proof}
        The DP-Inductive-Miner in Algorithm~\ref{algo:reject} is $\epsilon$-differentially private.
    \end{theorem}

    \begin{proof}
        By the post-processing theorem of differential privacy, it suffices that each operation that works on the sensitive event log is differentially private. These are the start and end count in Algorithm~\ref{algo:start_end} and the DP-DFR creation, fitness calculation, and rejection sampling in Algorithm~\ref{algo:reject}.

        In Algorithm~\ref{algo:start_end}, we count based on the event log in \texttt{cStart} and \texttt{cEnd} for each activity how often it is at the start and the beginning of each trace. Each activity in \texttt{cStart} and \texttt{cEnd} has a sensitivity of $1$, i.e. exchanging one trace increases each count by at most $1$. By the Laplace Mechanism (Definition \ref{def:Laplace}), adding Laplace noise to each $1$-sensitivity-bounded count with a scale of $\frac{4}{\eps_{\text{start\_end}}} = \frac{1}{0.25\eps_{\text{start\_end}}}$ is $0.25\eps_{\text{start\_end}}$-DP. As only two counts in either \texttt{cStart} and \texttt{cEnd} can increase by an exchanged trace, we can apply the parallel composition theorem \cite{privacybook} which means that noising all counts in either \texttt{cStart} and \texttt{cEnd} is $0.5\eps_{\text{start\_end}}$-DP. As we noise both \texttt{cStart} and \texttt{cEnd}, we apply the sequential composition theorem such that this process is $\eps_{\text{start\_end}}$-DP since $\eps_{\text{start\_end}} = 0.5\eps_{\text{start\_end}} + 0.5\eps_{\text{start\_end}}$.

        Algorithm \ref{algo:reject} utilizes the rejection sampler~\cite{privateCanidateSelection} (cf. Section~\ref{Sec:RejectionSampling}) which is $\eps$-DP with $\eps = 2\epsilon_1 + \epsilon_0$. Since we chose the number of activity pairs $n$ uniformly at random, the mechanism is not dependent on the rejection iteration. Thus, it remains to show that the fitness calculation is $\eps_1$-DP with $\eps_1 = r_1 \eps_1 + r_2 \eps_1 + r_3 \eps_1$ for some positive shares $r_1, r_2, r_3$ s.t. $r_1 + r_2 + r_3 = 1$. We notate $r_2 \eps_1 = \eps_{\text{start\_end}}$.

        Selecting the activity pair of \texttt{DFR} with the highest count using the Report Noisy Max is $0.5 r_1 \eps_1$-DP. As this process is repeated $n$ times, we apply sequential composition and rescale the privacy budget of the Report Noisy Max by $\sfrac{1}{n}$. Counting the frequency of \texttt{dfr} is $0.5 r_1 \eps_1$-DP with a similar argumentation as for Algorithm~\ref{algo:start_end}: The counting process of each count is $1$-sensitivity bounded, but here we apply sequential composition as each trace can alter all activity pairs. Thus, we have to scale the Laplace noise for each time we noise, i.e., by $n$. Counting the fitness works analogously and is $r_3 \eps_1$-DP, but here we have a sensitivity of $\sfrac{1}{|L|}$. Note that we assume that the event log size $L$ is public knowledge, which does not change if we exchange a trace in the event log. Thus, by the post-processing theorem, the fitness score calculation is $\eps_1$-DP, which concludes that creating a PST in Algorithm~\ref{algo:reject} is $\epsilon$-differentially private.
    \end{proof}

\section{Evaluation}
    \label{Sec:Evaluation}
    This section summarizes the evaluation results in terms of process model accuracy and trade-off between accuracy and privacy. We applied DPIM on 14 event logs from the BPI Challenges. Table \ref{tab:event_log_statistics} shows for each evaluated event log the amount of traces, trace variants, events, and unique activities.

    \begin{table}[h!]
        \centering
        \begin{tabular}{l|c|c|c|c}
            \toprule
            {BPI Challenge Event Log} & {Traces} & {Variants} & {Events} & {Activities} \\
            \midrule
            Closed Problems & 1,487 & 327 & 6,660 & 7 \\
            Domestic Declarations & 10,500 & 99 & 53,437 & 17 \\
            Incidents & 7,554 & 2,278 & 65,533 & 13 \\
            International Declarations & 6,449 & 753 & 72,151 & 34 \\
            Open Problems & 819 & 182 & 2,531 & 5 \\
            Prepaid Travel Costs & 2,099 & 202 & 18,246 & 29 \\
            Request for Payment & 6,886 & 89 & 36,796 & 19 \\
            Sepsis & 1,050 & 846 & 15,214 & 16 \\
            \bottomrule
        \end{tabular}
        \caption{Evaluation Event Log Statistics}
        \label{tab:event_log_statistics}
    \end{table}
    \begin{figure}[t]
        \vspace{-1em}
        \centering
        \includegraphics[width=0.475\textwidth]{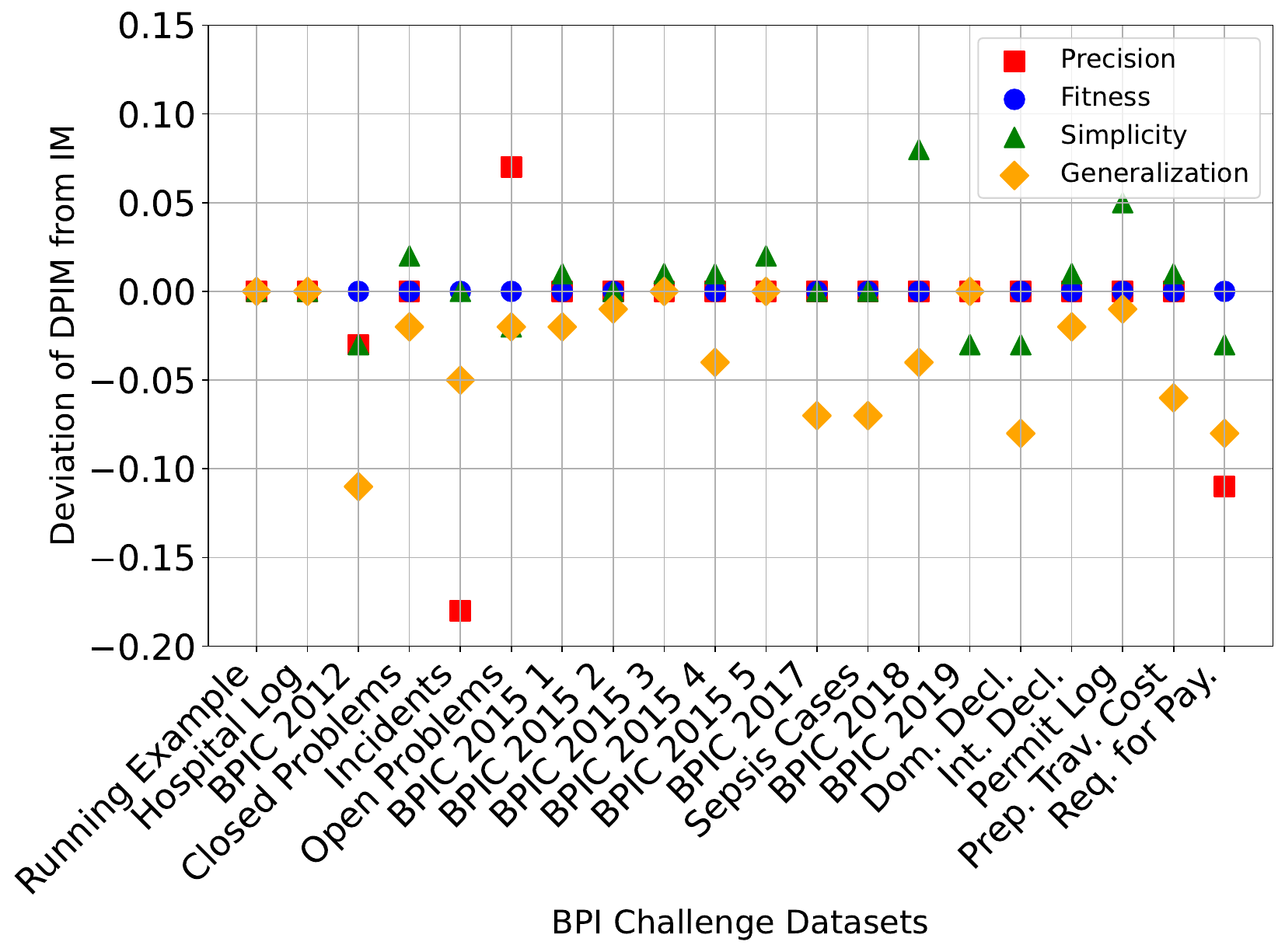}
        \caption{IM vs. the non-DP DPIM: Deviation of metrics}
        \label{fig:deviation_im_dpim}
        \vspace{-1em}
    \end{figure}
    \begin{figure*}[t]
        \centering
        \begin{minipage}[t]{0.215\textwidth}
        \centering
            \includegraphics[width=\textwidth]{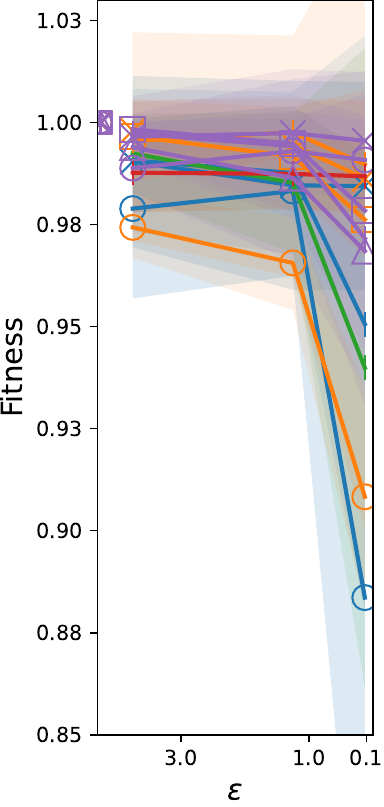}
        \end{minipage}
        \begin{minipage}[t]{0.215\textwidth}
        \centering
            \includegraphics[width=1\textwidth]{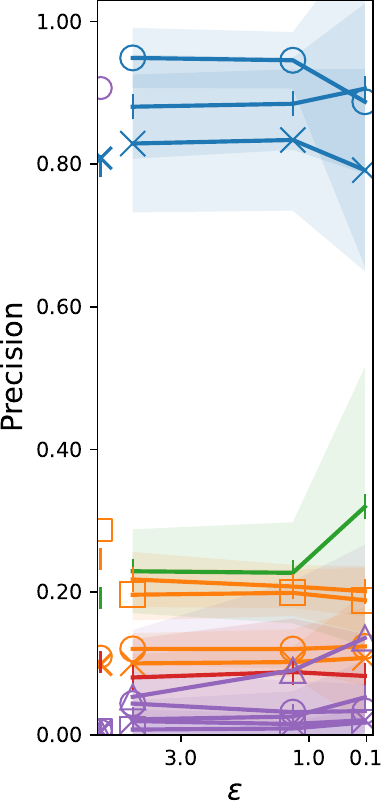}
        \end{minipage}
        \begin{minipage}[t]{0.215\textwidth}
        \centering
            \includegraphics[width=1\textwidth]{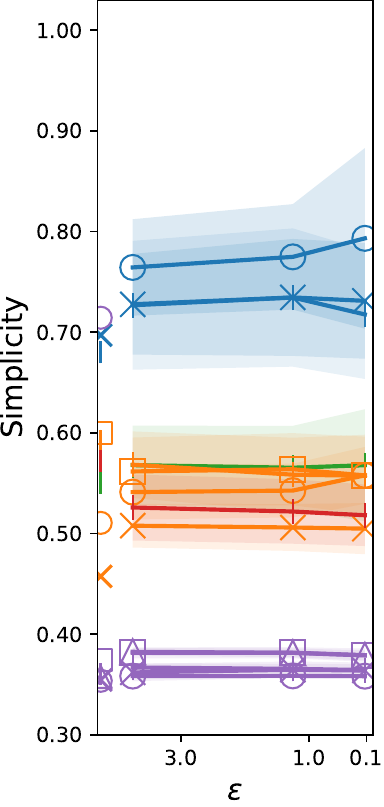}
        \end{minipage}
        \begin{minipage}[t]{0.215\textwidth}
        \centering
            \includegraphics[width=1\textwidth]{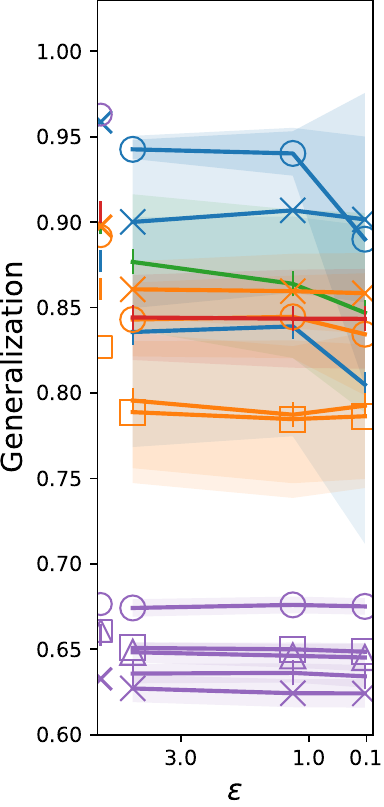}
        \end{minipage}\\
        \begin{minipage}{\textwidth}
        \centering
            \includegraphics[width=1\textwidth]{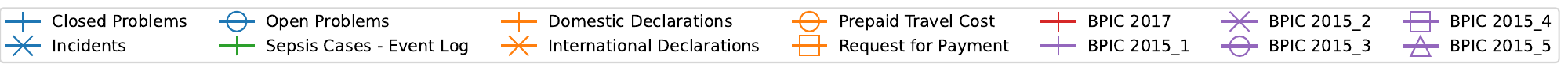}
        \end{minipage}
        \caption{Fitness, precision, simplicity, and generalization (higher is better) of DPIM for 8 benchmark event logs and 3 privacy parameters (lower $\varepsilon$ means stronger privacy). The dots at the Y-axis indicate the respective performance of the IM.}
        \label{fig:eval_generalization}
        \label{fig:eval_simplicity}
        \label{fig:eval_fitness}
        \label{fig:eval_precision}
    \end{figure*}

    \pparagraph{Evaluation Setup}
    We used PM4Py \cite{pm4py} to compute fitness, precision, simplicity, and generalization. As mentioned in Section \ref{Sec:Approach}, we modify some parts of PM4Pys cut detection to allow a privacy-preserving generation of a PST. As hyperparameters of Algorithm \ref{algo:reject}, we have chosen $r_1 = 0.65$, $r_2=0.25$, $r_3=0.1$, $\epsilon_0 = 0.01$, $\gamma = 0.01$, and $t = 0.95$.
    The lower and upper bound hyperparameters vary per event log. In our evaluation, we determined both bounds using the number of activity pairs in the directly-follows relation with a frequency count above 0. We subtracted $15$ from this frequency count for the lower bound and added $15$ for the upper bound. Both values were then rounded to the next by $5$ divisible number. Moreover, the lower bound is never smaller than the number of activities in the event log $L$, and the upper bound is always smaller than the square of the number of activities.

    \pparagraph{Evaluating Correctness}
    \label{Sec:PM_wo_DP}
    We ran DPIM in a non-DP mode (simulated by $\eps = 100{,}000$) on all event logs and compared the fitness, precision, simplicity, and generalization of the PSTs with the results of the Inductive Miner (IM). We found that DPIM approximates IM closely; nearly all metrics of the DPIM are within $\pm$ 0.10 of the respective IM values (Fig.~\ref{fig:deviation_im_dpim}). For instance on the 2013 Incidents event log, Figs. \ref{fig:sepsis_im} and \ref{fig:sepsis_nondp} show that process trees discovered by IM and DPIM are close.

    \begin{figure}[!ht]
        \centering
        \begin{subfigure}[b]{0.38\columnwidth}
            \centering
            \includegraphics[width=1\columnwidth]{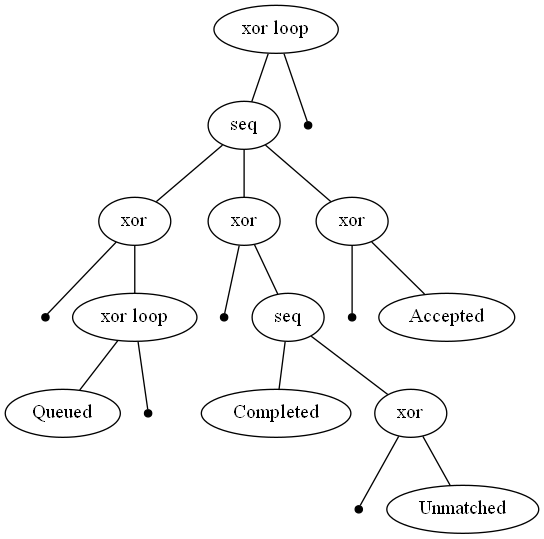}
            \caption{IM}
            \label{fig:sepsis_im}
        \end{subfigure}
        \begin{subfigure}[b]{0.38\columnwidth}
            \centering
            \includegraphics[width=1\columnwidth]{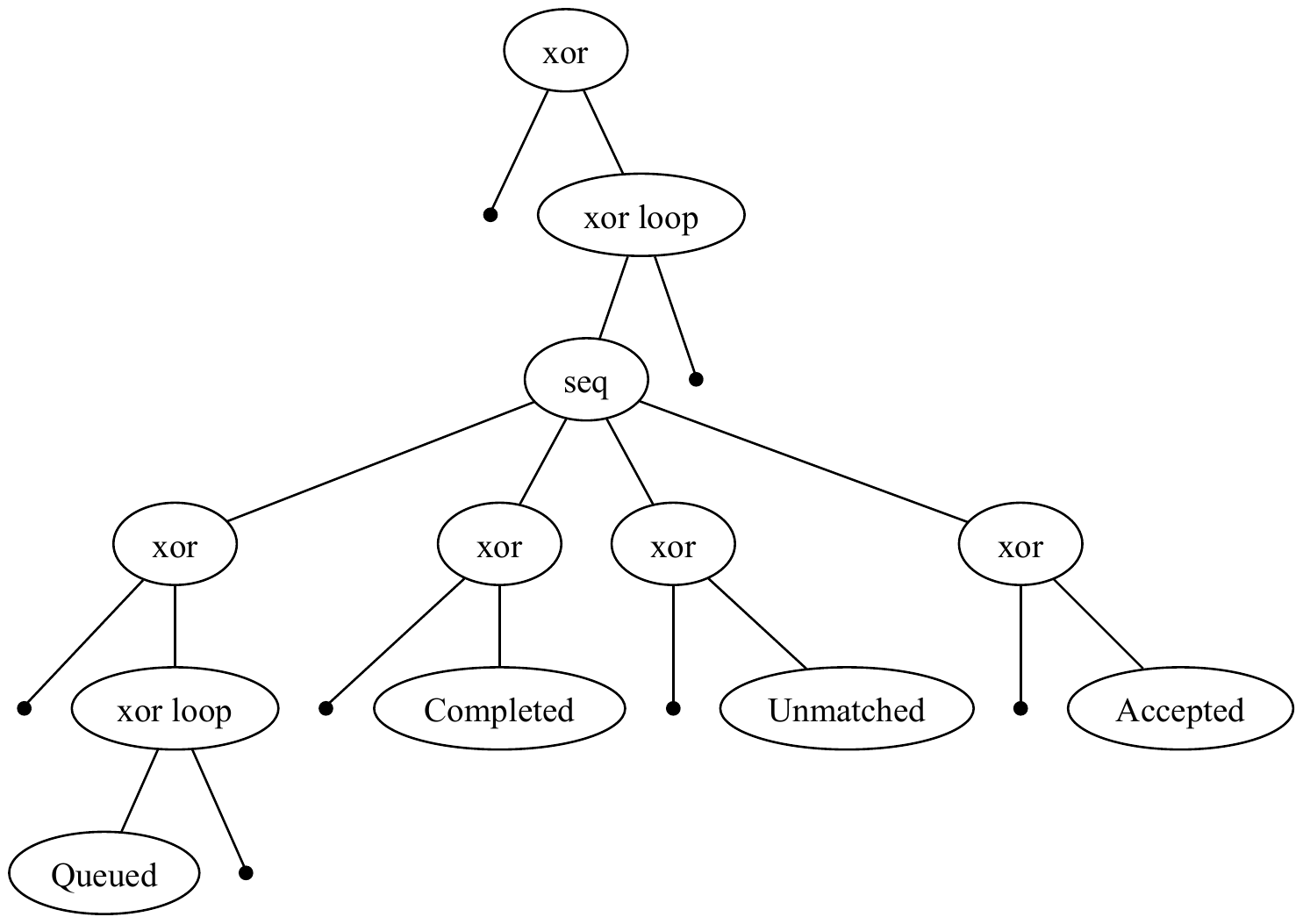}
            \caption{$\varepsilon$ = 100,000}
            \label{fig:sepsis_nondp}
        \end{subfigure}
        \begin{subfigure}[b]{0.21\columnwidth}
            \centering
            \includegraphics[width=1\columnwidth]{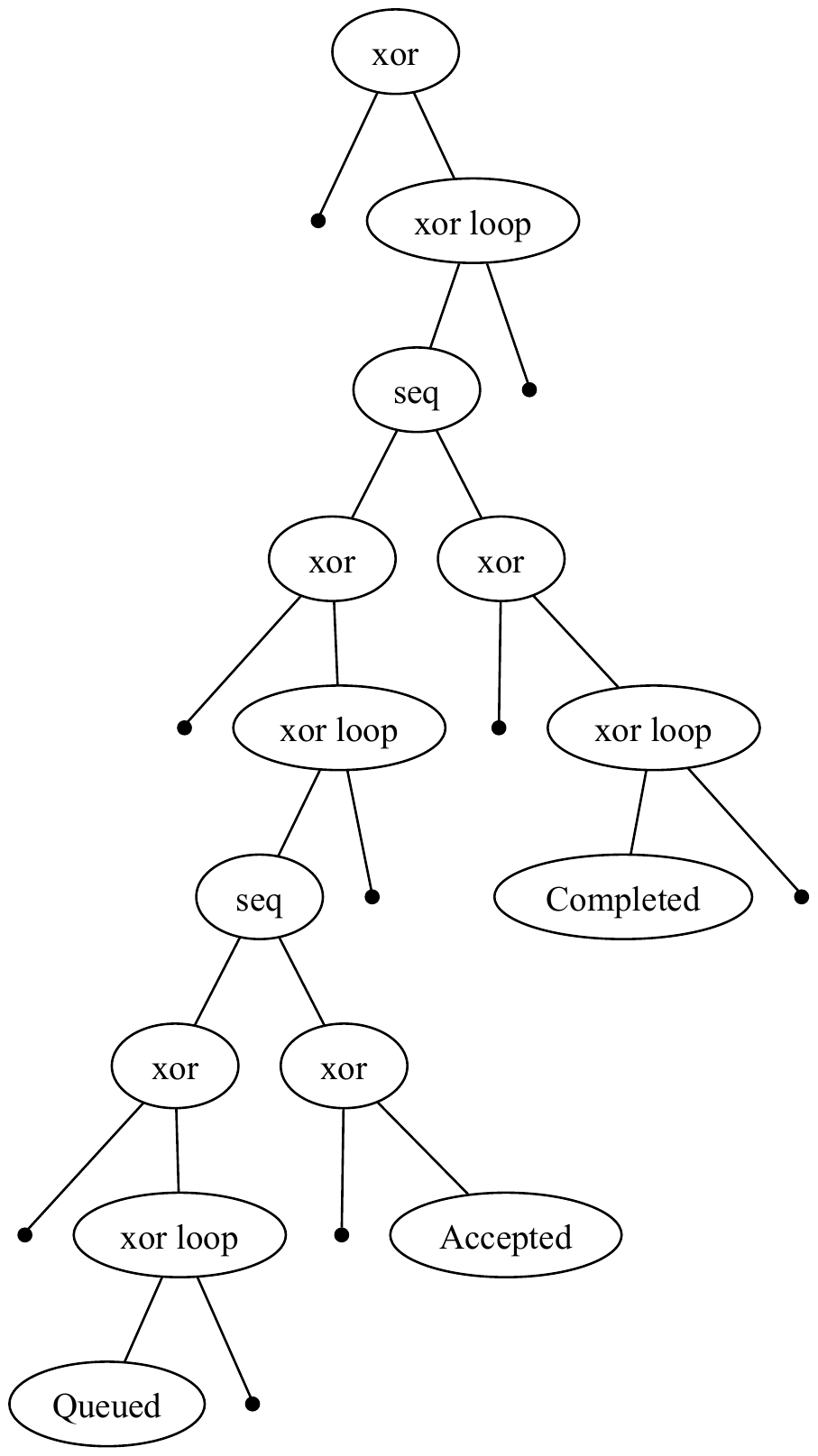}
            \caption{$\varepsilon$ = 0.125}
            \label{fig:sepsis_dpim}
        \end{subfigure}
        \caption{Comparing IM (a) and  DPIM (b+c) on \emph{Incidents} data.}
        \label{fig:deviation}
    \end{figure}

    \pparagraph{Evaluating Privacy Gain}
    \label{Sec:PM_w_DP}
    The strength of the DP guarantee is expressed with $\epsilon$ where a lower value means better privacy. By design, a higher $\epsilon$ introduces more noise, implicating how drastically a PST can be altered. Comparing Figs. \ref{fig:sepsis_nondp} and \ref{fig:sepsis_dpim} shows that the discovered PST changes significantly, hiding the exact nature of the underlying process. Thereby, the influence of single traces on the process model has been obscured, showing that the introduction of DP, as proven in Theorem~\ref{dp_proof}, alters the PST with limited utility loss (cf. Fig. \ref{fig:eval_generalization}).
    
    \pparagraph{Evaluating Utility Loss}
    \label{Sec:Tradeoff}
    Fig.~\ref{fig:eval_fitness} illustrates the fitness, precision, simplicity, and generalization of the 14 event logs of the IM and DPIM. The IM values are denoted as dots on the Y-axis. The privacy-preservation DPIM results are shown with varying $\epsilon$ values (3.75, 1.25, and 0.125). For \textbf{fitness}, most process models achieve a value of 0.95, where only the process model of the \textit{Open Problems} event log falls below 0.9. This robustness of fitness results from the rejection sampling method, as the DPIM is designed to preserve a high fitness to still be representative of the event log. For \textbf{precision}, the introduction of differential privacy increases the metric value for 10 out of 14 event logs. As the $\epsilon$ decreases, we see two effects. While some process models slightly decline in precision (\textit{Open Problems}) as the noised process model allows for more behavior not seen in the event log, some also increase (\textit{Closed Problems}). The latter can be explained by noise removing choices from the original process, thus allowing for fewer unseen process behaviors. \textbf{Simplicity} initially increases across all event logs compared to the IM. With a further increase in privacy, simplicity also further increases (\textit{Open Problems}, \textit{Prepaid Travel Cost}) while for most process models simplicity remains stable. The increase in simplicity comes from noise removing subtrees, leading to shorter and simpler process trees. \textbf{Generalization} declines across all event logs when introducing privacy, and mostly declines further when strengthening privacy.

\section{Conclusion}
    \label{Sec:Conclusion}
    This paper presented the differentially private discovery algorithm based on the Inductive Miner (IM) called \emph{DPIM}. We have proven $\eps$-DP of \emph{DPIM} and evaluated it on 14 real-world event logs. A comparison between the process models discovered by DPIM and the original IM shows that strong privacy can be given with a slight loss of utility. The choice of the budget $\epsilon$ quantifies the privacy guarantee. The loss of utility is quantified via the difference in quality metrics such as fitness, precision, simplicity, and generalization between the results of IM and DPIM.

\bibliographystyle{plain}
\bibliography{bibliography}

\end{document}